\documentclass[conference]{IEEEtran}
\IEEEoverridecommandlockouts
\usepackage{tabularx}
\usepackage{cite}
\usepackage{amsmath,amssymb,amsfonts,amsthm}
\let\oldnl\nl
\newcommand{\nonl}{\renewcommand{\nl}{\let\nl\oldnl}}
\usepackage{graphicx}
\usepackage{textcomp}
\usepackage{xcolor}
\usepackage[ruled,vlined,linesnumbered]{algorithm2e}
\usepackage{pgfplots}
\usepackage{tikz}
\usepgfplotslibrary{groupplots}
\pgfplotsset{compat=1.13}
\newlength\fwidth

\LinesNumbered
\usepackage{subcaption}

\SetKwRepeat{Do}{do}{while}%
\usepackage[colorlinks=true,  linkcolor=black, citecolor=blue, breaklinks=true, urlcolor=blue]{hyperref}
\def\BibTeX{{\rm B\kern-.05em{\sc i\kern-.025em b}\kern-.08em
    T\kern-.1667em\lower.7ex\hbox{E}\kern-.125emX}}
    
\newcommand{\E}{\mathbb{E}}                  
\newcommand{\bo}[1]{\mathbf{#1}}              
\newcommand{\bom}[1]{\boldsymbol{#1}}    
\newcommand{\tildeb}{\mathbf{c}}

\newcommand{\bmu}{\boldsymbol{\mu}}

\newcommand{\supp}{\mathsf{supp}}

\newcommand{\Gam}{\boldsymbol{\Gamma}}

\renewcommand{\Pi}{{\mathcal P}}
\newcommand{\hop}{\mathsf{H}}        

\newcommand{\iidsim}{\overset{iid}{\sim}}
\newcommand{\gfnc}{g}  
\newcommand{\ffnc}{f} 

\newcommand{\beq}{\begin{equation}}
\newcommand{\eeq}{\end{equation}}
\newcommand{\bmat}{\begin{pmatrix}}
\newcommand{\emat}{\end{pmatrix}}

\def\Gam{\mathbf{\Gamma}}

\newcommand{\R}{\mathbb{R}}

\newcommand{\X}{{\bf X}}
\newcommand{\SCM}{\mathbf{S}}

\newcommand{\A}{\mathbf{A}}

\newcommand{\Y}{\mathbf{Y}}

\newcommand{\B}{{\bo B}}
\renewcommand{\S}{\mathbf{S}} 

\newcommand{\sigmaw}{\sigma^2}

\newcommand{\C}{\mathbb{C}} 

\newcommand{\y}{\bo y}

\renewcommand{\b}{\mathbf{b}}

\newcommand{\x}{\bo x}

\renewcommand{\a}{{\bo a}}
\renewcommand{\c}{\mathbf{c}}

\newcommand{\M}{\bom \Sigma}

\DeclareMathOperator{\tr}{tr}
\DeclareMathOperator{\Tr}{tr}
\DeclareMathOperator{\diag}{diag}
\DeclareMathOperator{\cov}{cov}

\DeclareMathOperator{\Diag}{diag}

\newcommand{\bm}[1]{\mathbf{#1}}   
 \newcommand{\Hmat}{\mathsf{H}}

\newcommand{\gam}{\boldsymbol{\gamma}}   
  \makeatletter
\newcommand\notsotiny{\@setfontsize\notsotiny\@vipt\@viipt}
\makeatother
 \newcounter{ctheorem}
\newtheorem{theorem}[ctheorem]{Theorem}

\newcommand{\iter}{k}
\begin{document}

\title{Joint Activity Detection and Channel Estimation for Massive Random Access Using SBL and SCA
\thanks{The work of M. Esfandiari was supported by the Research Council of Finland under grant 368630, and D. Palomar in part by the Hong Kong GRF 16206123 research grant.}
}

\author{\IEEEauthorblockN{Esa Ollila\IEEEauthorrefmark{1}, Majdoddin Esfandiari\IEEEauthorrefmark{1}, and Daniel P. Palomar\IEEEauthorrefmark{2}}
\IEEEauthorblockA{\IEEEauthorrefmark{1}Department of Information and Communications Engineering,
Aalto University, Espoo, Finland\\
{\IEEEauthorrefmark{2}Department of Electronic and Computer Engineering, 
The Hong Kong University of Science and Technology, Hong Kong}\\
{Email: esa.ollila@aalto.fi, majdoddin.esfandiari@aalto.fi, palomar@ust.hk}} }

\maketitle

\begin{abstract}
In massive machine-type communication (mMTC) applications, a key challenge is 
joint device activity detection and channel estimation (JADCE) under grant-free random access, as a massive number of devices with sporadic traffic seek to connect to the base station. We address JADCE for massive random access using a covariance learning-based sparse Bayesian learning (SBL) approach. Specifically, we first use the successive convex approximation (SCA) framework to partially linearize the scaled negative log-likelihood function (LLF) of the data, then minimize it to estimate the sparse vector of devices' signal powers. After identifying active devices from these power estimates, empirical Bayesian estimation is used to obtain channel estimates. Simulation results demonstrate  the efficiency and performance superiority of the proposed CL-SCA method compared to other existing methods.       
\end{abstract}

\begin{IEEEkeywords}
Activity detection, channel estimation, sparse Bayesian learning (SBL), covariance learning, successive convex approximation (SCA)  
\end{IEEEkeywords}

\section{Introduction}
Joint device activity detection (AD) and channel estimation in grant-free random access is a key challenge in massive machine-type communication (mMTC) applications where a massive number of devices with sporadic data traffic wish to connect to the network in the uplink. Several studies have shown that covariance learning (CL-)based approaches can outperform conventional compressed sensing (CS-)based schemes \cite{haghighatshoar2018improved,chen2019covariance,fengler2021non,chen2021phase,marata2025activity,wang2025robust}.
Unlike traditional CS-based techniques, which require the number of active devices to be smaller than the pilot sequence length, covariance-based methods relax this constraint. In particular, \cite{fengler2021non,chen2021phase} demonstrate that covariance-based approaches can achieve comparable performance while using significantly shorter pilot sequences than their CS-based counterparts.
 
We consider an uplink, single-cell massive random access scenario with $N$ single-antenna machine-type devices (MTDs) communicating with a base station (BS) equipped with $M$ antennas.  For device identification, each device $n$ is preassigned a unique pilot (signature) sequence $\mathbf{a}_n = (a_{n1} \, \cdots \, a_{nL})^\top \in \mathbb{C}^{L}$ where $L$ is the pilot length.  Note that the pilot matrix $\A = (\mathbf{a}_1 \, \cdots \, \mathbf{a}_N) 
\in \C^{L\times N}$  collecting all pilot sequences as its column vectors is known at the BS. The device traffic is assumed to be sporadic, that is, in each coherence interval only $K \ll N$ devices are randomly active. 
In mMTC settings, the total number of devices is typically much larger than the pilot length, i.e., $N \gg L$. Let  $\alpha_n \in \{0, 1\}$~denote a binary indicator of device activity ($= 1$ when device is active and $= 0$ otherwise).  Let $\mathcal{M}$ denote the active user index set, i.e., $\mathcal{M} = \{n \in [\![ N ]\!] : \alpha_n = 1\}$ where $[\![ N ]\!]=\{1,\ldots,N\}$. Then, the received signal matrix  $\Y = (\y_1 \, \cdots \,\y_M) \in \C^{L \times M}$ over $L$ signal dimensions (symbols) and $M$ antennas can be expressed as  
\beq \label{eq:MMVmodel} 
\mathbf{Y}  = \sum_{n=1}^N \sqrt{\rho_n \beta_n}\alpha_n \a_n \mathbf{h}_n^{\top} + \mathbf{E} = \A \mathbf{X} + \mathbf{E},
\eeq
where $\rho_n$ is the $n^{\text{th}}$ device's uplink transmission power, $\beta_n$ is the large-scale fading component (LSFC) accounting for path-loss and shadowing, 
$\bm{h}_n \in \mathbb{C}^{M}$ represents the uplink channel vector for the $n^{\text{th}}$ device,  
$\mathbf{E} \in \C^{L \times M}$ is noise matrix whose elements $e_{lm}$ are i.i.d. with circular Gaussian distribution with variance $\sigma^2$ ($e_{lm} \sim \mathcal{C N}(0,\sigma^2)$ for $l \in [\![ L ]\!] $ and $m \in [\![ M ]\!] $), and 
$\X = (\tilde\x_1 \, \cdots \, \tilde\x_N)^\top \in \C^{N \times M}$ is the effective channel matrix whose (transposed) row vector 
\begin{align}
\tilde\x_{n} = \sqrt{\gamma_{n}} \mathbf{h}_{n}, \  n=1,\ldots, N,
\end{align}
is modeling the channel vector between the $n^{\text{th}}$ device and the BS,  and $\gamma_{n} = \alpha_n \rho_n \beta_n \geq 0$.  Since only $K \ll N$ devices are active, $\X$ is $K$-row-sparse, and thus $\Y$ in \eqref{eq:MMVmodel} follow the classic multiple measurement vectors  (MMV) model  \cite{duarte2011structured}.  The objective of AD is therefore to identify the indices of the active devices, i.e., the elements of the set $\mathcal{M}$, given the received signal~$\Y$, the pilot matrix~$\A$, and the noise power~$\sigma^2$. 

Under the uncorrelated Rayleigh small-scale fading and far-field assumptions, the channel coefficients can be treated as i.i.d. Gaussian,  both across antenna and symbol dimensions, i.e., $h_{nm} \iidsim  \mathcal {C N}(0,1)$ for $n \in [\![N]\!]$ and $m \in [\![M]\!]$. Since the elements of the noise matrix $\mathbf{E}$ are also i.i.d.,
the columns of $\Y$ for fixed $\gam$ can be modelled as i.i.d. random vectors drawn from circular Gaussian distribution, i.e., $\y_m \mid \gam  \sim \mathcal{C N}_{L}(\mathbf{0}, \M)$. The positive definite Hermitian (PDH) covariance matrix $\M=\cov(\y_m \mid \gam) \in \C^{L \times L}$ has the form 
\beq \label{eq:M}
\M = \A \Gam \A^\hop + \sigma^2 \mathbf{I}  =  \sum_{n=1}^N \gamma_n \a_n \a_n^\hop + \sigma^2 \mathbf{I}  
\eeq
where the signal covariance matrix $\Gam =  \diag(\gam)$ with $\gam = (\gamma_1\,\cdots\,\gamma_N)^\top \in \mathbb{R}^N_{\geq 0}$.   As already noted, since  only $K$ devices are active, 
 $\gam$ is a sparse vector of signal powers with only $K = | \mathcal M |$ non-zero elements.  

In this paper, we develop a CL-based method for joint activity detection and channel estimation (JADCE) for massive random access. We exploit the successive convex approximation (SCA) framework \cite{scutari2013decomposition} to first estimate the sparse vector of signal powers $\gam$ by minimizing the negative log-likelihood function (LLF) of the data $\Y$. From the estimated $\gam$, the elements of the active user index set $\mathcal M$ are then identified. Finally, empirical Bayesian estimation is employed to obtain uplink channel estimates of active devices. Simulation
results demonstrate that the proposed CL-SCA method is relatively fast and provides state-of-the-art JADCE performance among optimization-based competitors. 

\section{Covariance learning and SBL}

\subsection{Activity detection using covariance learning} 

Let  $\phi(\x ; \bmu, \M) =  | \M |^{-1}  \exp(- [\x-\bmu]^\hop \M^{-1} [\x-\bmu])$ denote the probability density function  (up to an irrelevant normalizing constant) of the circular Gaussian distribution with mean $\bmu$ and PDH covariance matrix $\M$. Assume $\sigma^2$ is known. Then, the scaled negative LLF of the data $\Y $ is 
\beq
\begin{aligned}  \label{eq:SML}
\ell( \boldsymbol{\gamma})   &= \frac{1}{M} \log \Big(\prod_{m=1}^M \phi(\y_m ; \mathbf{0}, \M) \Big) \\
 &= \tr(  (\A \Gam \A^\hop + \sigmaw \mathbf{I})^{-1} \SCM ) +  \log  |  \A \Gam \A^\hop + \sigmaw \mathbf{I} |
\end{aligned}
\eeq
where $\M$ is given by \eqref{eq:M}, $\SCM$ is  the  $L\times L$ sample covariance matrix (SCM), defined as 
\[
\SCM = \frac{1}{M} \sum_{m=1}^{M} \y_{m} \y_m^\hop =M^{-1} \mathbf{Y} \mathbf{Y}^\hop, 
\]
and $\tr(\cdot)$ and $| \cdot |$ denote the matrix trace and determinant.  We thus need to solve 
\begin{equation} 
  \begin{array}{ll}
  \underset{\gam}{\text{minimize}} \  \ell( \boldsymbol{\gamma})   \quad   \text{subject to}   \qquad \gam \ge \mathbf{0}, 
  \end{array}
  \label{eq:orig_formulation}
\end{equation}
Note that the first and second terms in the objective are convex and concave, respectively, and therefore the overall objective is nonconvex.
The active indices are obtained using the estimate $\hat{\gam}$ obtained by 
solving \eqref{eq:orig_formulation} 
and either of the rules:
\begin{itemize}
    \item The \emph{threshold-rule}:  
\beq \label{eq:threshold_rule}
\hat{\alpha}_n = \mathsf{1}{\{\hat \gamma_n \geq \gamma_{\text{thr}} \}} = \begin{cases} 1, & \text{if $\hat \gamma_n \geq \gamma_{\text{thr}}$ } \\ 0 & \text{otherwise} \end{cases} 
\eeq
where $\gamma_{\text{thr}}$ is a fixed threshold. 
\item  The \emph{top $K$-rule}:
\beq \label{eq:topKrule}
\hat \alpha_n = \mathsf{1}\!\left\{ \gamma_n \text{ is among the $K$ largest entries of } \hat{\gam} \right\}
\eeq
where $\mathsf{1}\{ \cdot \}$ is the indicator function and $K$ is given.
\end{itemize}
The latter requires knowledge (or an estimate) of the number of active devices $K= | \mathcal M|$ while the former requires calibrating a suitable threshold, either using a priori knowledge or via estimation. 
The threshold can be designed such that a desired trade-off between the probabilities of missed detection and false alarm is achieved.

\subsection{Joint activity detection and channel estimation via M\text{-}SBL} 

The M\text{-}SBL algorithm~\cite{wipf2007empirical}  is an iterative sparse signal reconstruction method for the MMV model. It also solves the CL problem in \eqref{eq:SML} and \eqref{eq:orig_formulation} to estimate $\gam$ via the EM algorithm. In our application context, M\text{-}SBL enables JADCE, since it provides not only an estimate of $\gam$, but also the posterior moments of the row-sparse source signal $\X$. The former is typically used to identify active devices via \eqref{eq:threshold_rule} or \eqref{eq:topKrule}, while the latter is used to estimate $\X$ (the effective channel matrix in our application context).
Under the stated assumptions,  the posterior distribution of $\X$ factorizes across antennas as 
$
p(\X \mid \Y, \gam)
= \prod_{m=1}^{M} \phi\!\big(\x_m;\, \bmu_m, \M_{x}\big),
$
where $\x_m \in \C^{N}$ denotes the $m^{\text{th}}$ column vector of $\X =(\x_1 \, \cdots \,\x_M)$. 
The posterior depends on   the conditional mean vectors  $\bmu_m \in \mathbb{C}^N$ of $\x_m$-s as well as  their joint posterior covariance matrix, $\M_x \in \mathbb{C}^{N \times N}$, defined by 
\begin{align*} 
&\E[\X \mid \Y, \gam]
= (\bmu_1 \, \cdots \, \bmu_M)  
= \Gam \A^\hop \M^{-1} \Y, \\ 
&\M_x = \cov(\x_m \mid \y_m, \gam)=  \Gam - \Gam \A^\hop \M^{-1} \A \Gam,
\end{align*} 
where $\M$ is defined by \eqref{eq:M}.
The M\text{-}SBL algorithm~\cite{wipf2007empirical} employs an empirical Bayesian (Type-II) approach to estimate the posterior distribution. 
In the Type-II framework, rather than placing a prior on the hyperparameters~$\gam$, one estimates them directly from the data. This is accomplished by minimizing the scaled negative LLF of $\Y$  given by $\ell(\gam)$ in~\eqref{eq:SML} subject to non-negativity constraint.
 After obtaining $\hat{\gam}$ by solving \eqref{eq:orig_formulation}, the posterior source distribution is computed via the plug-in principle as \(p(\X \mid \Y, \hat{\gam})\), and
\[
\hat{\X} = \hat{\Gam} \A^\hop \hat{\M}^{-1} \Y, 
\qquad
\hat{\M}_x = \hat{\Gam} - \hat{\Gam}\,\A^\hop \hat{\M}^{-1} \A \hat{\Gam},
\]
will serve as the empirical Bayes estimates of the posterior mean vectors and covariance matrix, respectively. Note that $\hat{\Gam} = \diag(\hat{\gam})$ and $\hat{\M} = \A \hat{\Gam} \A^{\hop} + \sigma^2 \mathbf{I}$. The algorithm for JADCE using M\text{-}SBL framework is given in \autoref{alg:SBL-JADCE}.

M-SBL minimizes the CL (i.e., Type-II likelihood) objective function $\ell(\gam)$ by constructing an EM algorithm that treats
$(\y_m,\x_m)$, $m=1,\ldots,M$, as the complete data.  Unfortunately, it is well-known that the developed EM algorithm suffers from slow convergence, making M\text{-}SBL  impractical for applications where $N$ is large. In this paper, we develop an efficient algorithm for solving \eqref{eq:orig_formulation} using the SCA framework \cite{scutari2013decomposition}. 
The proposed CL‑SCA algorithm converges significantly faster than the EM algorithm of M\text{-}SBL and also has smaller run times compared to the popular coordinate‑wise optimization (CWO) method \cite{fengler2021non}. Compared to a greedy CL-based approach such as CL matching pursuit (CL-MP) \cite[Algorithm~1]{marata2025activity} or CL-OMP \cite{ollila2024matching}, the proposed CL-SCA algorithm may have a higher runtime when $K$ is known and  $K$ is smaller than the number of iterations CL-SCA requires to terminate. This idealized scenario, which favors greedy methods like CL-MP, is examined in the simulation section.

\begin{algorithm}[!t]
\caption{Generic JADCE scheme based on the M\text{-}SBL framework \cite{wipf2007empirical}}\label{alg:SBL-JADCE}
\DontPrintSemicolon
\SetKwInOut{Input}{Input} 
\SetKwInOut{Output}{Output}
\SetKwInOut{Init}{Initialize}
\SetNlSkip{1em}
\SetInd{0.5em}{0.5em}
\Input{$\mathbf{Y}$, $\A$, $\sigma^2$, $K$ or $\gamma_{\text{thr}}$.} 

Solve \eqref{eq:orig_formulation}  to obtain $\hat \gam \in \mathbb{R}^N_{\geq 0}$ using an algorithm of your choice.

Compute estimate of activity indices $\hat{\boldsymbol{\alpha}}\! = \!(\alpha_1\,\cdots\,\alpha_N)^{\top} \!\in \! \{0,1\}^N$ using either the threshold rule in  \eqref{eq:threshold_rule} or the $K$-top rule in \eqref{eq:topKrule}.   

Prune the signal power estimates as  $\hat{\gam}  \gets \hat{\gam} \odot \hat{\boldsymbol{\alpha}}$, where $\odot$ denotes the Hadamard (element-wise) product. 

Compute effective channel estimate 
$
\hat{\X} = \hat{\Gam} \A^\hop \hat{\M}^{-1} \Y, 
$
where $\hat{\Gam} = \diag(\hat{\gam})$ and $\hat{\M} = \A \hat{\Gam} \A^{\hop} + \sigma^2 \mathbf{I}$. 

 \Output{$\hat{\boldsymbol{\alpha}}$ (or support $\hat{\mathcal M} = \supp(\hat{\boldsymbol{\alpha}})$), $\hat \X$, $\hat{\gam}$.} 
\end{algorithm}

\section{CL-SCA algorithm} 

 Under the  SCA framework, each agent in the optimization focuses on a single variable, $\gamma_1, \gamma_2, \dots, \gamma_N$, which are updated iteratively and in parallel. In other words, all $N$ variables are updated simultaneously at each iteration.

We begin with expressing the objective function in \eqref{eq:SML} as the sum of a convex term and a nonconvex term (to be linearized), as follows:
\beq
\label{eq:conv+nonconv}
\ell(\gam) =  \gfnc(\gam) +  \ffnc(\gam),
\eeq
where $g : \mathbb{R}^N_{\geq 0} \to \R$ is a convex continuously differentiable in $\mathbb{R}^N_{\geq 0}$ 
and $f : \mathbb{R}^N_{\geq 0} \to \R$ is a nonconvex function, 
 defined as 
 \begin{align}
   \gfnc(\gam) &= \Tr ( (\bm{A}\Diag(\gam)\A^\hop + \sigma^2\,\bm{I})^{-1}\,\S ) \\ 
   \ffnc(\gam) &= \log |\bm{A}\Diag(\gam)\bm{A}^\Hmat + \sigma^2\,\bm{I} |. 
 \end{align}
We define $\M_{\backslash i} = \sum_{n \neq i}  \gamma_n   \a_n \a_n^\hop + \sigmaw \mathbf{I}  = \M - \gamma_i \a_i \a_i^{\hop}$ as the covariance matrix of  $\y_m$-s when the contribution of the $i^{\text{th}}$ source signal is removed, and $\tildeb_i =  \M^{-1}_{\backslash i}  \a_i$ and $\b_i =  \M^{-1}  \a_i$ for $i=1,\ldots, N$. The convex term can therefore be written as
\begin{align}
  \gfnc(\gam)
  &= \Tr\left(\big(\gamma_i\,\bm{a}_i\bm{a}_i^\Hmat + \bm{\Sigma}_{\backslash i}\big)^{-1}\,\bm{S} \right) \notag \\
  &= \Tr(\bm{\Sigma}_{\backslash i}^{-1}\bm{S}) - \frac{\gamma_i}{1 + \gamma_i\,\bm{a}_i^\Hmat \bm{\Sigma}_{\backslash i}^{-1} \bm{a}_i}\Tr \big(\bm{\Sigma}_{\backslash i}^{-1}\bm{a}_i\bm{a}_i^\Hmat\bm{\Sigma}_{\backslash i}^{-1}\bm{S}\big) \notag\\
  &= \Tr\big(\bm{\Sigma}_{\backslash i}^{-1}\bm{S}\big) - \frac{\gamma_i \bm{a}_i^\Hmat\bm{\Sigma}_{\backslash i}^{-1}\bm{S}\bm{\Sigma}_{\backslash i}^{-1}\bm{a}_i}{1 + \gamma_i\,\bm{a}_i^\Hmat \bm{\Sigma}_{\backslash i}^{-1} \bm{a}_i} \notag  \\
  &= \Tr\big(\bm{\Sigma}_{\backslash i}^{-1}\bm{S}\big) - \frac{\gamma_i \c_i^\hop \S \c_i}{1 + \gamma_i\,\a_i^\Hmat \c_i}\label{eq:f_cvx_gammai}
\end{align}
where we have used the Sherman-Morrison formula\footnote{
$(\A + \mathbf{u} \mathbf{v}^\hop)^{-1}=\A^{-1} - \A^{-1} \mathbf{u} \mathbf{v}^\hop \A^{-1}/(1 + \mathbf{v}^\hop \A^{-1} \mathbf{u})$}. At iteration $k+1$, by following the SCA framework and linearizing $f(\gam)$, the $N$ functions  to be minimized by
 the $N$ agents in parallel  are given as
\begin{align}    
\label{eq:tilde_f}
\tilde{\ell}_i(\gamma_i \mid \gam^{\iter}) 
&= \gfnc\big(\gamma_i, \gam_{\backslash i}^{\iter} \big) + \nabla_{\gamma_i} \ffnc\big(\gam^{\iter} \big)\,(\gamma_i - \gamma_i^{\iter}) \nonumber \\
&= - \frac{\gamma_i [\c_i^{\iter}]^\Hmat \S \c_i^{\iter}}{1 + \gamma_i\,\a_i^{\hop} \c_i^{\iter}} + \gamma_i\,\bm{a}_i^{\hop} \b_i^{\iter} + \text{const}\, ,
\end{align}
where $i \in [\![ N ]\!]$, $\gam^{k} =(\gamma_1^{k},\ldots ,\gamma_N^{k})^\top$, 
$\c_i^{\iter} = 
(\M^{\iter} - \gamma_i^{\iter}  \a_i \a_i^{\hop})^{-1}  \a_i$, $\M^k= \A \diag(\gam^k) \A^{\hop} + \sigma^2 \mathbf{I}$, and $\b_i^{\iter} 
= \big(\M^{\iter} \big)^{-1} \a_i $.
 The solution to this optimization problem is given below. 

\begin{theorem}  \label{th:SCA} For $i \in \{1,\ldots, N\}$, the minimizer of the  convex function $\tilde{\ell}_i(\gamma_i \mid \gam^{\iter})$ is 
\begin{align}  
\hat{\gamma_i}(\gam^{k}) &= \arg \min_{\gamma_i \geq 0} \,  \tilde{\ell}_i(\gamma_i \mid \gam^{\iter}) \notag  \\ 
&= \left[ \gamma_i^{\iter}  + \sqrt{\frac{ [\b_i^{\iter}]^{\hop} \SCM \b_i^{\iter}}{ \big(\a_i^\hop \b_i^{\iter} \big)^3 }}  -   \frac{1}{\a_i^\hop \b_i^{\iter}}\right]_+  \label{eq:SCA_optim_gammai}
\end{align} 
where $[a]_+=\max(a,0)$.
\end{theorem} 

\begin{proof} We provide the proof in 
in Appendix~\ref{app:th:SCA}.
\end{proof} 

Finally, $\gam^{k+1}$ is obtained via a smoothing operation, which is necessary to guarantee convergence in accordance with the SCA framework \cite{scutari2013decomposition}, as 
\begin{align}
\gam^{\iter+1} &= \big(1 - \eta^k \big)\gam^{\iter} + \eta^k\,\hat{\gam}\left(\gam^k\right) \nonumber \\ 
&=  \gam^{\iter} + \eta^k\, \left(\hat{\gam}(\gam^k) - \gam^{\iter} \right)  \label{eq:SCA_smoothing2} 
\end{align}
where $\hat \gam(\gam^k) = (\hat \gamma_1(\gam^k)\, \cdots\, \hat \gamma_N(\gam^k))^\top$ is formed by the solutions of \eqref{eq:SCA_optim_gammai}, while the step-size sequence $\eta^k \in (0,1]$ is introduced to control the length of the update along
the direction $\mathbf{d}^k = \hat{\gam}(\gam^k) - \gam^k$. It   can be conveniently chosen as a diminishing step-size, verifying 
 $\sum_{k=0}^{\infty} \eta^k = \infty$ and $\sum_{k=0}^\infty (\eta^k)^2 < + \infty$ \cite{scutari2013decomposition,scutari2018parallel}. We use $\eta^k = \eta^{k-1}\,\big(1 - \epsilon\,\eta^{k-1}\big)$, where $\epsilon \in (0,1)$ is a given constant, and $\eta^0 <1/\epsilon$. We set $\eta^0=0.99$ and $\epsilon=0.05$ as default values and use $\mathbf{\gam}^0 = \mathbf{0}$ as initial start. The CL-SCA algorithm is tabulated in \autoref{alg:SCA}. As a termination criterion, we use 
\beq \label{eq:term_CLSCA}
\| \mathbf{d}^k \|_2 = \| \hat \gam(\gam^k) - \gam^k\|_2 \leq \delta
\eeq
for an appropriate convergence threshold $\delta$  (default $\delta = 10^{-3}$),  as this directly tests how close we are to a fixed point of the SCA mapping, i.e., to stationarity. We wish to point out that all the elements are calculated in parallel and that the SCA algorithm is guaranteed to converge to the stationary point of the original problem.  The code is available at \url{https://github.com/esollila/CL-SCA}. 

\begin{algorithm}[!h]
 \caption{\textsf{CL-SCA} algorithm}\label{alg:SCA}
\DontPrintSemicolon
\SetKwInOut{Input}{Input} 
\SetKwInOut{Output}{Output}
\SetKwInOut{Init}{Initialize}
\SetNlSkip{1em}
\SetInd{0.5em}{0.5em}
\Input{$\mathbf{S}$, $\A$, $\sigma^2$, step-size sequence $\{\eta^k\}$
} 
\Init{$I_{\max}=50$, $\gam^0 = \mathbf{0}$,  $\delta= 10^{-3}$} 
 
 \BlankLine 
\For{$k =0,1,\ldots,I_{\max}-1$}{

\BlankLine 
Compute $\M^k= \A \diag(\gam^k) \A^{\hop} + \sigma^2 \mathbf{I}$, and $\B^k = \big( \M^k \big)^{-1} \A = (\b_1^{\iter} \, \cdots \,\b_N^{\iter}) $. 

For all $i \in \{1,\ldots, N\}$ solve in parallel: 
\[ 
\hat{\gamma_i}(\gam^{k}) = \left[ \gamma_i^{\iter}  + \sqrt{\frac{ [\b_i^{\iter}]^{\hop} \SCM \b_i^{\iter}}{ \big(\a_i^\hop \b_i^{\iter} \big)^3 }}  -   \frac{1}{\a_i^\hop \b_i^{\iter}}\right]_+ 
\]

 \BlankLine
$\mathbf{d}^k = \hat{\gam}(\gam^k) - \gam^k$
\BlankLine
$ 
 \gam^{\iter+1} =\gam^{\iter} + \eta^k\,\mathbf{d}^k 
$
 \BlankLine
\If{  $\| \mathbf{d}^k  \|_2 < \delta$}{\nonl \textsf{break} } 
 
 }
 \Output{$\hat{\boldsymbol{\gamma}}=\gam^{\iter+1}$} 
\end{algorithm}

The computational cost of steps~2  of CL-SCA is  $\mathcal{O}(N L^2) +    \mathcal{O}\!\left(L^{3} + L^{2} N\right)$ while the cost of step~3 is 
$\mathcal{O}(NL^2)$.  Computational complexity of Step 4-6 are  $\mathcal{O}(N)$. Since $N \geq L$ in our application, the complexity  of one iteration of CL-SCA algorithm is $\mathcal{O}(N L^2)$.    This is the same complexity per iteration as CWO \cite[Algorithm~1]{fengler2021non} and CL-MP \cite[Algorithm~1]{marata2025activity}.  

  The input to CL-based approaches is  $\mathbf{S}$ (the SCM) with the computational complexity of $\mathcal{O}\!\left(M L^2\right)$. 
Thus, the overall complexities of CL-SCA and CWO are  $\mathcal{O}(N L^2 I_{\text{iters}} + M L^2)$, where $I_{\text{iters}}$ is the number of iterations until convergence. Note that CL-MP is a greedy CL-based approach, terminated after $K$ iterations (when $K$ is known), and hence its overall complexity is  $\mathcal{O}(N L^2 K + ML^2)$. Thus CL-MP provides faster computations only in the ideal case when $K$ is known and $K<I_{\text{iters}}$.

\section{Simulation study } 

\subsection{Set-up} 

\noindent{\bf Performance metrics:}  Let $\hat{\mathcal M}$ and $\mathcal{M}$ denote the estimated and actual active user index sets, respectively. To evaluate the activity detection performance, we use  the probability of missed detection ($P_{\text{MD}}$):
$$
P_{\mathrm{MD}}
\triangleq
\mathbb{E}\!\left[ 
\frac{|\mathcal{M} \setminus \hat{\mathcal{M}}|}{|\mathcal{M}|}
\right]
$$
where the expectation is taken over the channel realizations and noise, and $\mathcal{M} \setminus \hat{\mathcal{M}}$ represents the set of elements in $\mathcal{M}$ that are not in $\hat{\mathcal{M}}$.  
We consider the case where the number of active devices $|\mathcal{M}|=K$ is known.  Accordingly, we will use the top-$K$ rule in \eqref{eq:topKrule} in \autoref{alg:SBL-JADCE}.
To investigate the  channel estimation performance, we use the normalized mean square error (NMSE):
\[
\text{NMSE} = \mathbb{E} \left[ \frac{  \| \hat \X - \X \|_{\text{F}}^2}{  \|  \X \|_{\text{F}}^2 }  \right]
\]
where $\| \cdot \|_{\text{F}}$ denotes the Frobenius matrix norm. 

\noindent {\bf Compared methods:} To evaluate the performance of the JADCE framework (\autoref{alg:SBL-JADCE}), we compare four methods used to solve \eqref{eq:orig_formulation} in step~1 of \autoref{alg:SBL-JADCE}, namely the proposed CL-SCA method (\autoref{alg:SCA}), the original EM algorithm of M-SBL \cite{wipf2007empirical}, CWO \cite[Algorithm~1]{fengler2021non}, and CL-MP \cite[Algorithm~1]{marata2025activity}. We stress here that step~4 of \autoref{alg:SBL-JADCE} is used to obtain empirical channel estimates for all the methods tested.

\noindent {\bf Setting:}  Normalized Bernoulli pilot sequences with unit power per symbol are used, i.e., $a_{ln} \in \{ \pm \frac{1}{\sqrt{2}} \pm  \jmath \frac{1}{\sqrt{2}} \}$, and  $\|\a_n\|_2^2=L$.
 The support $\mathcal M$ is randomly chosen from ${1,\ldots,N}$ without replacement for each Monte-Carlo (MC) trial.
The number of MC trials is $10000$. The number of MTDs is  $N = 300$. The LSFCs ($\beta_n$-s in \eqref{eq:MMVmodel}) are uniformly distributed between \([-15,0]\) in dB scale. 
 
\subsection{Discussion of  results} 

\begin{figure}[!t]
\begin{subfigure}{\columnwidth}
\hspace{3pt}
\centerline{\begin{tikzpicture}

\definecolor{crimson2143940}{RGB}{214,39,40}
\definecolor{darkgray176}{RGB}{176,176,176}
\definecolor{darkorange25512714}{RGB}{255,127,14}
\definecolor{forestgreen4416044}{RGB}{44,160,44}
\definecolor{steelblue31119180}{RGB}{31,119,180}

\begin{axis}[
axis line style={semithick},
width=0.9\textwidth,
height=6.8cm,
ylabel style={yshift=-3pt, font=\footnotesize} ,
xlabel style={yshift=3pt,font=\footnotesize},
tick label style={font=\tiny} , 
title style={font=\footnotesize},
ylabel style={yshift=-2pt},
  ytick distance=0.2, 
log basis y={10},
tick pos=left,
x grid style={darkgray176},
xmajorgrids,
xmin=17, xmax=83,
xminorgrids,
xtick style={color=black},
y grid style={darkgray176},
ylabel={Probability of missdetection},
ymajorgrids,
ymin=9.60023595175981e-06, ymax=0.176233168486847,
yminorgrids,
ymode=log,
ytick style={color=black},
ytick={1e-07,1e-06,1e-05,0.0001,0.001,0.01,0.1,1,10},
yticklabels={
 \( \displaystyle {10^{-7}}\),
  \(\displaystyle {10^{-6}}\),
  \(\displaystyle {10^{-5}}\),
  \(\displaystyle {10^{-4}}\),
  \(\displaystyle {10^{-3}}\),
  \(\displaystyle {10^{-2}}\),
  \(\displaystyle {10^{-1}}\),
  \(\displaystyle {10^{0}}\),
  \(\displaystyle {10^{1}}\)
},
legend style={at={(0,0)}, anchor=south west, legend cell align=left, align=left, draw=white!15!black,font=\tiny,opacity=0.99, row sep=-3.5pt,inner sep=1pt},
legend image post style={xscale=0.6,yscale=0.75},
legend columns=6,
transpose legend,
]
\addplot [thick, steelblue31119180, mark=*, mark size=1.9, mark options={solid}]
table {%
20 0.0240850000000013
30 0.0075299999999998
40 0.00241000000000002
50 0.000789999999999998
60 0.00022
70 8e-05
80 1.5e-05
};
\addlegendentry{SCA $K=20$}
\addplot [thick, steelblue31119180,  mark=square*, mark size=1.9, mark options={solid}]
table {%
20 0.0434000000000039
30 0.0205666666666658
40 0.0100099999999997
50 0.00459666666666667
60 0.00202333333333335
70 0.000979999999999998
80 0.000476666666666666
};
\addlegendentry{$K=30$}

\addplot [thick, steelblue31119180, mark=triangle*, mark size=2.5, mark options={solid}]
table {%
20 0.0712350000000005
30 0.0415999999999994
40 0.0256750000000027
50 0.0161175000000014
60 0.0101100000000004
70 0.00651999999999978
80 0.00411249999999988
};
\addlegendentry{$K=40$}
\addplot [thick, darkorange25512714,  dashed, mark=o, mark size=1.9, mark options={fill=none,solid}]
table {%
20 0.0239300000000013
30 0.00750999999999981
40 0.00238000000000002
50 0.000769999999999998
60 0.00022
70 8e-05
80 1.5e-05
};
\addlegendentry{CWO $K=20$}
\addplot [thick, darkorange25512714, dashed, mark=square, mark size=1.9, mark options={fill=none,solid}]
table {%
20 0.0433000000000039
30 0.0204166666666658
40 0.00993999999999973
50 0.00457333333333334
60 0.00198333333333335
70 0.000969999999999998
80 0.000476666666666666
};
\addlegendentry{$K=30$}
\addplot [thick, darkorange25512714,  dashed, mark=triangle, mark size=1.9, mark options={fill=none,solid}]
table {%
20 0.0711400000000005
30 0.0413774999999993
40 0.0255225000000027
50 0.0159800000000014
60 0.0099950000000004
70 0.00644249999999976
80 0.00406499999999988
};
\addlegendentry{$K=40$}
\addplot [semithick, forestgreen4416044, mark=*, mark size=1.9, mark options={fill=none,solid}]
table {%
20 0.021620000000001
30 0.00637999999999987
40 0.00192500000000001
50 0.000684999999999998
60 0.000225
70 8.5e-05
80 3e-05
};
\addlegendentry{CL-MP $K=20$}

\addplot [semithick, forestgreen4416044,  mark=square*, mark size=1.9, mark options={fill=none,solid}]
table {%
20 0.0408900000000034
30 0.0188699999999993
40 0.00880999999999979
50 0.00406666666666669
60 0.00198666666666668
70 0.000843333333333331
80 0.000499999999999999
};
\addlegendentry{$K=30$}
\addplot [semithick, forestgreen4416044, mark=triangle*, mark size=2.5, mark options={fill=none,solid}]
table {%
20 0.0714425000000003
30 0.0403799999999994
40 0.0246275000000025
50 0.0155425000000012
60 0.00950500000000031
70 0.00624249999999977
80 0.00403499999999988
};
\addlegendentry{$K=40$}
\addplot [semithick, crimson2143940, mark = o, mark size=1.9, mark options={fill=none,solid}] 
table {%
20 0.0234650000000013
30 0.00726499999999982
40 0.00222000000000002
50 0.000769999999999998
60 0.000205
70 8e-05
80 1.5e-05
};
\addlegendentry{EM $K=20$}
\addplot [semithick, crimson2143940,  mark=square, mark size=1.9, mark options={fill=none,solid}]
table {%
20 0.0427300000000038
30 0.0200333333333325
40 0.00969666666666641
50 0.00441333333333335
60 0.00191000000000001
70 0.000909999999999998
80 0.000449999999999999
};
\addlegendentry{$K=30$}
\addplot [semithick, crimson2143940,  mark=triangle, mark size=2.5, mark options={fill=none,solid}] 
table {%
20 0.0707050000000007
30 0.0410999999999993
40 0.0253175000000027
50 0.0157950000000014
60 0.00977250000000037
70 0.00634999999999976
80 0.00396749999999989
};
\addlegendentry{$K=40$}
\end{axis}

\end{tikzpicture}}
\centerline{\begin{tikzpicture}

\definecolor{crimson2143940}{RGB}{214,39,40}
\definecolor{darkgray176}{RGB}{176,176,176}
\definecolor{darkorange25512714}{RGB}{255,127,14}
\definecolor{forestgreen4416044}{RGB}{44,160,44}
\definecolor{steelblue31119180}{RGB}{31,119,180}

\begin{axis}[
axis line style={semithick},
width=0.9\textwidth,
height=6.8cm,
ylabel style={yshift=-3pt, font=\tiny} ,
xlabel style={yshift=3pt, font=\tiny} ,
tick label style={font=\tiny} , 
title style={font=\footnotesize},
ylabel style={yshift=-2pt},
  ytick distance=0.2, 
ylabel style={font=\footnotesize} ,
log basis y={10},
tick pos=left,
x grid style={darkgray176},
tick label style={/pgf/number format/assume math mode},
xmajorgrids,
xmin=17, xmax=83,
xminorgrids,
xtick style={color=black},
y grid style={darkgray176},
ylabel={Probability of missdetection},
xlabel style={font=\footnotesize},
ymajorgrids,
yminorgrids,
ytick style={color=black},
ymode=log,
ytick={0.02,0.05,0.1,0.2,0.3},    
yticklabels={$2\cdot10^{-2}$, $5\cdot10^{-2}$, $10^{-1}$,$2\cdot10^{-1}$,$3\cdot10^{-1}$},
label = {$L=30$},
legend style={at={(0,0)}, anchor=south west, legend cell align=left, align=left, draw=white!15!black,font=\tiny,opacity=0.99, row sep=-3pt,inner sep=1pt},
legend image post style={xscale=0.6,yscale=0.75},
legend columns=6,
transpose legend,
]

\addplot [semithick, steelblue31119180, mark=*, mark size=1.9, mark options={solid}]
table {%
20 0.125789999999999
30 0.0821349999999989
40 0.0579350000000028
50 0.0417600000000043
60 0.0300600000000024
70 0.0225750000000012
80 0.0159450000000001
};
\addlegendentry{SCA $K=20$}

\addplot [thick, steelblue31119180,  mark=square*, mark size=1.9, mark options={solid}]
table {%
20 0.187473333333352
30 0.139880000000009
40 0.110130000000006
50 0.0891800000000085
60 0.0738200000000079
70 0.0619833333333394
80 0.0528566666666717
};
\addlegendentry{$K=30$}

\addplot [thick, steelblue31119180, mark=triangle*, mark size=2.4, mark options={solid}]
table {%
20 0.251209999999999
30 0.200122499999999
40 0.167692500000001
50 0.1443575
60 0.1262175
70 0.112895000000002
80 0.100322500000004
};
\addlegendentry{$K=40$}
\addplot [thick, darkorange25512714,  dashed, mark=o, mark size=1.9, mark options={fill=none,solid}]
table {%
20 0.125554999999999
30 0.0819799999999989
40 0.057645000000003
50 0.0415200000000043
60 0.0298850000000024
70 0.0222250000000012
80 0.0156700000000001
};
\addlegendentry{CWO $K=20$}
\addplot [thick, darkorange25512714, dashed, mark=square, mark size=1.9, mark options={fill=none,solid}]
table {%
20 0.187243333333352
30 0.139580000000009
40 0.109766666666673
50 0.0887766666666752
60 0.0733633333333412
70 0.0615666666666726
80 0.0523466666666717
};
\addlegendentry{$K=30$}
\addplot [thick, darkorange25512714,  dashed, mark=triangle, mark size=1.9, mark options={fill=none,solid}]
table {%
20 0.2510425
30 0.199802499999999
40 0.167287500000001
50 0.14411
60 0.1258675
70 0.112462500000002
80 0.0998450000000045
};
\addlegendentry{$K=40$}

\addplot [semithick, forestgreen4416044, mark=*, mark size=1.9, mark options={fill=none,solid}]
table {%
20 0.124834999999999
30 0.078829999999999
40 0.0544750000000039
50 0.0385750000000036
60 0.0270450000000017
70 0.0200250000000007
80 0.0143299999999998
};
\addlegendentry{CL-MP $K=20$}
\addplot [semithick, forestgreen4416044,  mark=square*, mark size=1.9, mark options={fill=none,solid}]
table {%
20 0.199136666666683
30 0.146326666666676
40 0.113180000000006
50 0.0916766666666746
60 0.0756900000000072
70 0.0630133333333391
80 0.0528466666666716
};
\addlegendentry{$K=30$}
\addplot [semithick, forestgreen4416044, mark=triangle*, mark size=2.5, mark options={fill=none,solid}]
table {%
20 0.278042500000001
30 0.220442499999999
40 0.184515
50 0.159612500000001
60 0.14104
70 0.125972500000001
80 0.113530000000002
};
\addlegendentry{$K=40$}
\addplot [semithick, crimson2143940, mark = o, mark size=1.9, mark options={fill=none,solid}]
table {%
20 0.124459999999999
30 0.080924999999999
40 0.0572250000000032
50 0.0410050000000042
60 0.0297100000000023
70 0.0219350000000011
80 0.0157900000000001
};
\addlegendentry{EM $K=20$}
\addplot [semithick, crimson2143940,  mark=square, mark size=1.9, mark options={fill=none,solid}]
table {%
20 0.186616666666685
30 0.139526666666675
40 0.110246666666673
50 0.0900900000000085
60 0.0753666666666749
70 0.0642800000000064
80 0.0551900000000051
};
\addlegendentry{$K=30$}
\addplot [semithick, crimson2143940,  mark=triangle, mark size=2.5, mark options={fill=none,solid}]
table {%
20 0.250272499999999
30 0.200387499999999
40 0.168617500000001
50 0.14667
60 0.1293675
70 0.116985000000001
80 0.105620000000003
};
\addlegendentry{$K=40$}
\end{axis}

\end{tikzpicture}}
\centerline{\begin{tikzpicture}

\definecolor{crimson2143940}{RGB}{214,39,40}
\definecolor{darkgray176}{RGB}{176,176,176}
\definecolor{darkorange25512714}{RGB}{255,127,14}
\definecolor{forestgreen4416044}{RGB}{44,160,44}
\definecolor{steelblue31119180}{RGB}{31,119,180}

\begin{axis}[
axis line style={semithick},
width=0.9\textwidth,
height=6.8cm,
ylabel style={yshift=-3pt, font=\footnotesize} ,
xlabel style={yshift=3pt,font=\footnotesize},
tick label style={font=\tiny} , 
title style={font=\footnotesize},
ylabel style={yshift=-2pt},
  ytick distance=0.2, 
log basis y={10},
ytick={0.1,0.2,0.3,0.4,0.5},    
yticklabels={$10^{-1}$, $2\cdot10^{-1}$, $3 \cdot 10^{-1}$,$4\cdot10^{-2}$,$5\cdot10^{-1}$},
tick pos=left,
x grid style={darkgray176},
xlabel={$M$ (number of antennas)},
xmajorgrids,
xmin=17, xmax=83,
xminorgrids,
xtick style={color=black},
y grid style={darkgray176},
ylabel={Probability of missdetection},
ymajorgrids,
yminorgrids,
ymode=log,
ytick style={color=black},
legend style={at={(0,0)}, anchor=south west, legend cell align=left, align=left, draw=white!15!black,font=\tiny,opacity=0.99, row sep=-3.5pt,inner sep=1pt},
legend image post style={xscale=0.62,yscale=0.77},
legend columns=6,
transpose legend,
]
\addplot [semithick, steelblue31119180, mark=*, mark size=1.9, mark options={solid}]
table {%
20 0.271695
30 0.217400000000006
40 0.180610000000009
50 0.155265000000006
60 0.135505000000001
70 0.120734999999999
80 0.107914999999998
};
\addlegendentry{SCA $K=20$}

\addplot [thick, steelblue31119180,  mark=square*, mark size=1.9, mark options={solid}]
table {%
20 0.356660000000041
30 0.301773333333353
40 0.264136666666669
50 0.238023333333333
60 0.21761333333334
70 0.201290000000017
80 0.186693333333352
};
\addlegendentry{$K=30$}

\addplot [thick, steelblue31119180, mark=triangle*, mark size=2.4, mark options={solid}]
table {%
20 0.419080000000001
30 0.364144999999999
40 0.329182499999998
50 0.303425000000002
60 0.282250000000002
70 0.266642500000001
80 0.2517775
};
\addlegendentry{$K=40$}
\addplot [thick, darkorange25512714,  dashed, mark=o, mark size=1.9, mark options={fill=none,solid}]
table {%
20 0.27117
30 0.216770000000006
40 0.179975000000009
50 0.154560000000006
60 0.134620000000001
70 0.119779999999999
80 0.107099999999998
};
\addlegendentry{CWO $K=20$}
\addplot [thick, darkorange25512714, dashed, mark=square, mark size=1.9, mark options={fill=none,solid}]
table {%
20 0.356300000000041
30 0.301186666666687
40 0.263523333333335
50 0.237379999999999
60 0.216596666666674
70 0.200043333333351
80 0.185360000000018
};
\addlegendentry{$K=30$}
\addplot [thick, darkorange25512714,  dashed, mark=triangle, mark size=1.9, mark options={fill=none,solid}]
table {%
20 0.419975000000001
30 0.364327499999999
40 0.328797499999998
50 0.302725000000002
60 0.281177500000002
70 0.265375000000001
80 0.2506975
};
\addlegendentry{$K=40$}

\addplot [semithick, forestgreen4416044, mark=*, mark size=1.9, mark options={fill=none,solid}]
table {%
20 0.28817
30 0.227870000000005
40 0.188380000000008
50 0.161620000000005
60 0.139765000000001
70 0.124004999999999
80 0.109624999999997
};
\addlegendentry{CL-MP $K=20$}
\addplot [semithick, forestgreen4416044,  mark=square*, mark size=1.9, mark options={fill=none,solid}]
table {%
20 0.395536666666708
30 0.334100000000034
40 0.295383333333349
50 0.267806666666671
60 0.24664
70 0.230676666666669
80 0.216193333333341
};
\addlegendentry{$K=30$}
\addplot [semithick, forestgreen4416044, mark=triangle*, mark size=2.5, mark options={fill=none,solid}]
table {%
20 0.482875000000001
30 0.42418
40 0.389365
50 0.365359999999998
60 0.347122500000001
70 0.334384999999998
80 0.322889999999998
};
\addlegendentry{$K=40$}
\addplot [semithick, crimson2143940, mark = o, mark size=1.9, mark options={fill=none,solid}] 
table {%
20 0.27074
30 0.218785000000006
40 0.18477500000001
50 0.161240000000007
60 0.144280000000003
70 0.131660000000001
80 0.121054999999998
};
\addlegendentry{EM $K=20$}
\addplot [semithick, crimson2143940,  mark=square, mark size=1.9, mark options={fill=none,solid}]
table {%
20 0.355896666666708
30 0.303346666666687
40 0.269196666666671
50 0.244849999999999
60 0.227170000000003
70 0.213526666666677
80 0.201503333333351
};
\addlegendentry{$K=30$}
\addplot [semithick, crimson2143940,  mark=triangle, mark size=2.5, mark options={fill=none,solid}] 
table {%
20 0.417022500000003
30 0.364854999999999
40 0.331527499999998
50 0.308527500000001
60 0.289380000000002
70 0.275885000000002
80 0.263747500000002
};
\addlegendentry{$K=40$}

\end{axis}

\end{tikzpicture}}
\end{subfigure}
\caption{Probability of missed detection  vs. $M$ for different numbers of active devices ($K$) when $L=20,30, $ and $L=50$ from bottom to top, respectively.} \label{fig:AD1}
\end{figure}
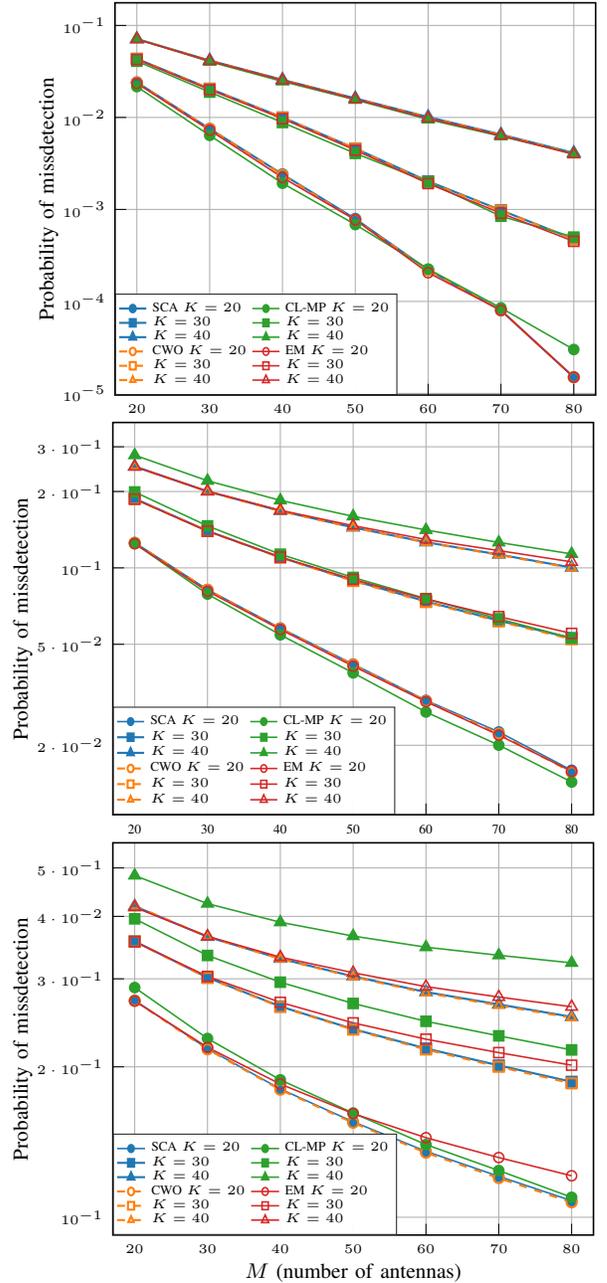 

We investigate the effects of $L$, $M$, and $K$ on the performance of JADCE. Fig.~\ref{fig:AD1} shows the probability of missed detection as a function of the number of antennas $M$, for different numbers of active devices $K$ and pilot lengths $L=20,30, $ and $L=50$ from bottom to top, respectively. When $L$ is small (e.g., $L=20$ in the bottom panel), it can be observed that CL-SCA and CWO are the best performing AD methods across all values of $K$. 

We also observe that while greedy CL-MP performs well when $K$ is small, its performance deteriorates as $K$ increases for $L=20$ and $L=30$. This study shows that CL‑SCA and CWO achieve very similar missed‑detection performance, while the EM algorithm of M-SBL performs slightly worse, with noticeable performance gaps for short pilot lengths (e.g., $L=20$).  The average running times for all the methods tested, for pilot lengths $L=20$ and $L=50$, are shown in Fig.~\ref{fig:CPU}. These results show that the proposed CL‑SCA method has a clear run time advantage over CWO, often being roughly twice as fast for the pilot lengths considered. The original EM implementation of M‑SBL is approximately one order of magnitude slower than CL-SCA, primarily because it requires several hundred iterations to converge. CL-MP is the computationally fastest method, as it is a greedy method that runs for exactly  $K$ iterations.

\begin{figure}[!t]
\centerline{\begin{tikzpicture}

\definecolor{crimson2143940}{RGB}{214,39,40}
\definecolor{darkgray176}{RGB}{176,176,176}
\definecolor{darkorange25512714}{RGB}{255,127,14}
\definecolor{forestgreen4416044}{RGB}{44,160,44}
\definecolor{steelblue31119180}{RGB}{31,119,180}

\begin{axis}[
width= 0.4\columnwidth, 
height= 4.1cm, 
 y tick label style={font=\notsotiny} , 
 x tick label style={font=\footnotesize} , 
scale only axis,
bar shift auto,
tick pos=left,
ylabel style={yshift=-3pt, font=\footnotesize} ,
xlabel style={yshift=3pt, font=\footnotesize} ,
x grid style={darkgray176},
xlabel={$K$ (number of active devices)},
xmin=-0.54, xmax=2.54,
xtick style={color=black},
xtick={0,1,2},
xticklabels={20,30,40},
y grid style={darkgray176},
ylabel={Running time [s]},
             yticklabel style={%
                 /pgf/number format/.cd,
                     fixed,
                     fixed zerofill,
                     precision=2,
                     },
ymajorgrids,
title = {(a) $L=20$},
title style={yshift=-5pt, font=\footnotesize} ,
ymin=0, ymax=0.1005693720439402,
ytick style={color=black},
legend style={at={(0,1)}, anchor=north west, legend cell align=left, align=left, draw=white!15!black,font=\tiny,opacity=0.89, row sep=-1.5pt,inner sep=1pt},
legend image post style={xscale=0.54,yscale=0.73},
legend columns=6,
]
\draw[draw=none,fill=steelblue31119180] (axis cs:-0.4,0) rectangle (axis cs:-0.2,0.00925378135222543);
\addlegendimage{area legend, draw=none, fill=steelblue31119180}
\addlegendentry{SCA}

\draw[draw=none,fill=steelblue31119180] (axis cs:0.6,0) rectangle (axis cs:0.8,0.0102936363790384);
\draw[draw=none,fill=steelblue31119180] (axis cs:1.6,0) rectangle (axis cs:1.8,0.01057401837352);
\draw[draw=none,fill=darkorange25512714] (axis cs:-0.2,0) rectangle (axis cs:-2.77555756156289e-17,0.0173838727932756);
\addlegendimage{area legend, draw=none, fill=darkorange25512714}
\addlegendentry{CWO}

\draw[draw=none,fill=darkorange25512714] (axis cs:0.8,0) rectangle (axis cs:1,0.0225201869317863);
\draw[draw=none,fill=darkorange25512714] (axis cs:1.8,0) rectangle (axis cs:2,0.0258311132197069);
\draw[draw=none,fill=forestgreen4416044] (axis cs:-2.77555756156289e-17,0) rectangle (axis cs:0.2,0.00205933255521315);
\addlegendimage{area legend, draw=none, fill=forestgreen4416044}
\addlegendentry{CL-MP}

\draw[draw=none,fill=forestgreen4416044] (axis cs:1,0) rectangle (axis cs:1.2,0.0032217520030681);
\draw[draw=none,fill=forestgreen4416044] (axis cs:2,0) rectangle (axis cs:2.2,0.0041242294733213);
\draw[draw=none,fill=crimson2143940] (axis cs:0.2,0) rectangle (axis cs:0.4,0.0857423786827563);
\addlegendimage{area legend, draw=none, fill=crimson2143940}
\addlegendentry{EM}

\draw[draw=none,fill=crimson2143940] (axis cs:1.2,0) rectangle (axis cs:1.4,0.0900660686132764);
\draw[draw=none,fill=crimson2143940] (axis cs:2.2,0) rectangle (axis cs:2.4,0.0890218709360941);
\end{axis}

\end{tikzpicture}
\begin{tikzpicture}

\definecolor{crimson2143940}{RGB}{214,39,40}
\definecolor{darkgray176}{RGB}{176,176,176}
\definecolor{darkorange25512714}{RGB}{255,127,14}
\definecolor{forestgreen4416044}{RGB}{44,160,44}
\definecolor{steelblue31119180}{RGB}{31,119,180}

\begin{axis}[
width= 0.4\columnwidth, 
height= 4.1cm, 
 y tick label style={font=\notsotiny} , 
 x tick label style={font=\footnotesize} , 
scale only axis,
bar shift auto,
tick pos=left,
ylabel style={yshift=-3pt, font=\footnotesize} ,
xlabel style={yshift=3pt, font=\footnotesize} ,
x grid style={darkgray176},
xlabel={$K$ (number of active devices)},
xmin=-0.54, xmax=2.54,
xtick style={color=black},
xtick={0,1,2},
xticklabels={20,30,40},
             yticklabel style={%
                 /pgf/number format/.cd,
                     fixed,
                     fixed zerofill,
                     precision=2,
                     },
y grid style={darkgray176},
ymajorgrids,
title = {(b) $L=50$},
title style={yshift=-5pt, font=\footnotesize} ,
ymin=0, ymax=0.403,
ytick style={color=black},
legend style={at={(0,1)}, anchor=north west, legend cell align=left, align=left, draw=white!15!black,font=\tiny,opacity=0.89, row sep=-1.5pt,inner sep=1pt},
legend image post style={xscale=0.54,yscale=0.73},
legend columns=6,
]
\draw[draw=none,fill=steelblue31119180] (axis cs:-0.4,0) rectangle (axis cs:-0.2,0.0160815896962503);
\addlegendimage{area legend, draw=none, fill=steelblue31119180}
\addlegendentry{SCA}

\draw[draw=none,fill=steelblue31119180] (axis cs:0.6,0) rectangle (axis cs:0.8,0.0172023235081828);
\draw[draw=none,fill=steelblue31119180] (axis cs:1.6,0) rectangle (axis cs:1.8,0.0192481257167379);
\draw[draw=none,fill=darkorange25512714] (axis cs:-0.2,0) rectangle (axis cs:-2.77555756156289e-17,0.0280762739162693);
\addlegendimage{area legend, draw=none, fill=darkorange25512714}
\addlegendentry{CWO}

\draw[draw=none,fill=darkorange25512714] (axis cs:0.8,0) rectangle (axis cs:1,0.031618528309052);
\draw[draw=none,fill=darkorange25512714] (axis cs:1.8,0) rectangle (axis cs:2,0.0367708938700879);
\draw[draw=none,fill=forestgreen4416044] (axis cs:-2.77555756156289e-17,0) rectangle (axis cs:0.2,0.00627632606787208);
\addlegendimage{area legend, draw=none, fill=forestgreen4416044}
\addlegendentry{CL-MP}

\draw[draw=none,fill=forestgreen4416044] (axis cs:1,0) rectangle (axis cs:1.2,0.00940416258622759);
\draw[draw=none,fill=forestgreen4416044] (axis cs:2,0) rectangle (axis cs:2.2,0.0124603366222565);
\draw[draw=none,fill=crimson2143940] (axis cs:0.2,0) rectangle (axis cs:0.4,0.323906862421327);
\addlegendimage{area legend, draw=none, fill=crimson2143940}
\addlegendentry{EM}

\draw[draw=none,fill=crimson2143940] (axis cs:1.2,0) rectangle (axis cs:1.4,0.346311635551615);
\draw[draw=none,fill=crimson2143940] (axis cs:2.2,0) rectangle (axis cs:2.4,0.361479112333073);
\end{axis}

\end{tikzpicture}}
\caption{Average running time vs. the number of active devices $K$ when $L=20, $ and $L=50$.}
\label{fig:CPU}
\end{figure}
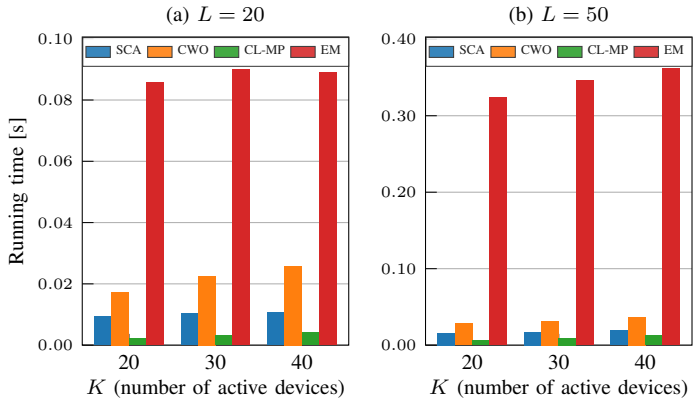

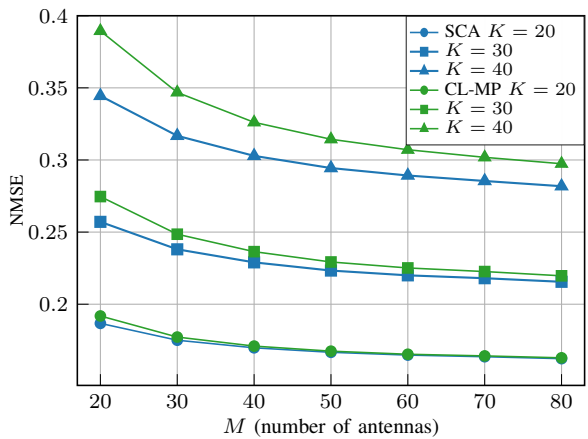
\begin{figure}[!t]
\centerline{\begin{tikzpicture}

\definecolor{darkgray176}{RGB}{176,176,176}
\definecolor{forestgreen4416044}{RGB}{44,160,44}
\definecolor{steelblue31119180}{RGB}{31,119,180}

\begin{axis}[
axis line style={semithick},
width=0.94\columnwidth,
height=6.5cm,
ylabel style={yshift=-3pt, font=\scriptsize} ,
xlabel style={yshift=3pt, font=\footnotesize} ,
y tick label style={font=\footnotesize} , 
x tick label style={font=\footnotesize}, 
title style={font=\footnotesize},
ylabel style={yshift=-2pt},
  ytick distance=0.2, 
tick pos=left,
x grid style={darkgray176},
xlabel={$M$ (number of antennas)},
xmajorgrids,
xmin=17, xmax=83,
xminorgrids,
xtick style={color=black},
y grid style={darkgray176},
ylabel={NMSE},
ymajorgrids,
ymin=0.143710783012801, ymax=0.4,
ytick={0.2,0.25,0.3,0.35,0.4}, 
yminorgrids,
ytick style={color=black},
legend style={at={(1,1)}, anchor=north east, legend cell align=left, align=left, draw=white!15!black,font=\scriptsize,opacity=0.99, row sep=-3pt,inner sep=1pt},
legend image post style={xscale=0.63,yscale=0.77},
legend columns=6,
transpose legend,
]

\addplot [semithick, steelblue31119180, mark=*, mark size=1.9, mark options={solid}]
table {%
20 0.192646808944643
30 0.181101317922742
40 0.175534146371809
50 0.17210221397645
60 0.169822792904592
70 0.168349778245402
80 0.16677817607422
};
\addlegendentry{SCA $K=20$}

\addplot [thick, steelblue31119180,  mark=square*, mark size=1.9, mark options={solid}]
table {%
20 0.258402621639376
30 0.241067197503111
40 0.232573247599855
50 0.22715785877979
60 0.22387183341594
70 0.221811532212556
80 0.219232778579905
};
\addlegendentry{$K=30$}
\addplot [thick, steelblue31119180, mark=triangle*, mark size=2.4, mark options={solid}]
table {%
20 0.340790720364443
30 0.316222208931332
40 0.303615300363223
50 0.29598911775393
60 0.291205816625278
70 0.287661737586766
80 0.284225362062344
};
\addlegendentry{$K=40$}
\addplot [semithick, forestgreen4416044, mark=*, mark size=1.9, mark options={fill=none,solid}]
table {%
20 0.191854548419014
30 0.177232390756661
40 0.170974412265132
50 0.167448674157034
60 0.165310703149715
70 0.164168902933728
80 0.162835681333449
};
\addlegendentry{CL-MP $K=20$}

\addplot [semithick, forestgreen4416044,  mark=square*, mark size=1.9, mark options={fill=none,solid}]
table {%
20 0.274636032858597
30 0.248532468758153
40 0.23643879039714
50 0.229277430841471
60 0.225139974404685
70 0.222641038443056
80 0.219693730210649
};
\addlegendentry{$K=30$}

\addplot [semithick, forestgreen4416044, mark=triangle*, mark size=2.5, mark options={fill=none,solid}]
table {%
20 0.389446209846489
30 0.34682906949594
40 0.32609344904037
50 0.314353992075277
60 0.307092994188484
70 0.30182809478466
80 0.297398423440109
};
\addlegendentry{$K=40$}

\end{axis}

\end{tikzpicture}}
\caption{Channel estimation NMSE vs. the number of antennas $M$ for different numbers of active devices $K$ when pilot length is $L=30$.} \label{fig:NMSE}
\end{figure} 

Fig.~\ref{fig:NMSE} shows the channel estimation performance for a pilot length of $L=30$. Here, we only compare  CL-SCA and CL-MP.  
As can be observed, only in the settings where CL-MP provides slightly better AD performance ($K=20$ and $L=30$ as shown in the middle panel of Fig.~\ref{fig:AD1}), its channel estimation performance is slightly better than that of CL-SCA. For $K=40$, CL-SCA provides clearly improved performance while for $K=30$ the differences in channel estimation performance are marginal. We have not included the CWO or EM approaches in this figure, as their curves overlap with that of CL-SCA.

\section{Conclusion} 
We have studied joint activity detection and channel estimation for massive random access in uplink mMTC scenarios using the CL-based SBL framework. Based on successive convex approximation framework, we have developed the CL-SCA algorithm by decomposing the scaled negative LLF into a convex and a nonconvex term, linearizing the latter, minimizing the resultant objective function, and updating all power variables in parallel to achieve an efficient iterative method. Simulation results show that CL-SCA often achieves similar or superior performance in terms of the probability of missed detection compared to CWO, the EM algorithm of M-SBL, and CL-MP. While the average running time of greedy CL-MP is generally smaller than that of CL-SCA  in the ideal case when $K$ is known and $K<I_{\text{iters}}$, its channel estimation performance is typically worse. It can be concluded from simulation results that the proposed CL-SCA algorithm is fastest among optimization based methods  and  provides top activity detection and channel estimation performance.   


\clearpage
\appendix[Proof of Theorem~\ref{th:SCA}]\label{app:th:SCA}
First we note that 
\begin{align}
 \nabla_{\gamma_i} \, \ffnc(\gam) &=  \frac{\partial\log|\M |}{\partial\gamma_i} 
 = \mathrm{tr}\bigg({\M^{-1} \frac{\partial \M}{\partial\gamma_i}}\bigg) 
  =   \a_i^\hop   \M^{-1} \a_i \notag \\ 
  &= \a_i^\hop \b_i  \label{eq:partial_gammai_ncvx} \\ 
   \nabla_{\gamma_i}\gfnc(\gam)  &= - \frac{ \c_i^\hop \S \c_i}{(1 + \gamma_i\,\a_i^\Hmat \c_i)^2} 
   \label{eq:partial_gammai_ncvx2}
\end{align} 
where \eqref{eq:partial_gammai_ncvx2} follows by taking the derivative of \eqref{eq:f_cvx_gammai}. 
With the help of \eqref{eq:f_cvx_gammai} and \eqref{eq:partial_gammai_ncvx} we can write \eqref{eq:tilde_f} compactly as 
\begin{align}  \label{eq_tildef_i}
\tilde{\ell}_i(\gamma_i \mid \gam^{\iter}) 
&= - \frac{\gamma_i [\c_i^{\iter}]^\Hmat \S \c_i^{\iter}}{1 + \gamma_i\,\a_i^{\hop} \c_i^{\iter}} + \gamma_i\,\bm{a}_i^{\hop} \b_i^{\iter}  
\end{align}
where $\b_i^{\iter}$ and $\c_i^{\iter} $, defined  by 
\begin{align} 
\b_i^{\iter} &= \big(\M^{\iter} \big)^{-1} \a_i 
= (\A  \diag(\gam^{\iter}) \A^\hop + \sigma^{2}\mathbf{I})^{-1} \a_i \\
\c_i^{\iter} &=\big(\M^{\iter}_{\setminus i} \big)^{-1} \a_i  = (\M^{\iter} - \gamma_i^{\iter}  \a_i \a_i^{\hop})^{-1}  \a_i
\end{align} 
 are the values of vectors $\b_i$ and $\c_i$   evaluated at  $\gam = \gam^{\iter}$ and $\M^{\iter} = \A \diag(\gam^{\iter}) \A^{\hop} + \sigma^2 \mathbf{I}$. 
Note also that we have ignored additive constants in \eqref{eq_tildef_i}  that do not depend on the optimization variable $\gamma_i$.

Taking the derivative of $\tilde{\ell}_i(\gamma_i \mid \gam^{\iter})$ in \eqref{eq_tildef_i} w.r.t. variable $\gamma_i$ and setting it to zero yields the equation
\begin{align*}
0 &=  \nabla_{\gamma_i} \tilde{\ell}_i(\gamma_i \mid \gam^{\iter} ) = - \frac{[\c_i^{\iter}]^\Hmat \S \c_i^{\iter}}{ (1 + \gamma_i\,\a_i^{\hop} \c_i^{\iter})^2} + \bm{a}_i^{\hop} \b_i^{\iter} , 
\end{align*} 
which has a single root: 
\beq \label{eq:gamma_i_update}
\tilde \gamma_i(\gam^k)  = \sqrt{\frac{ [\c_i^{\iter}]^{\hop} \S \c_i^{\iter}}{\bm{a}_i^{\hop} \b_i^{\iter} (\bm{a}_i^\Hmat \c_i^{\iter})^2}} - \frac{1}{\bm{a}_i^\Hmat \c_i^{\iter}} . 
\eeq
In SCA approach, each of the powers are updated {\it simultaneously} (in parallel). 
This expression is not convenient due to the term $\c_i^{\iter}$ as it requires the computation of all the inverses $\big(\bm{\Sigma}_{\backslash i}^k\big)^{-1}$ for each coordinate $i=1,\dots,M$.

Applying the Sherman-Morrison formula  to $\M_{\setminus i}^k = \M^k - \gamma_i^k \a_i \a_i^{\hop}$ we notice that vectors $\tildeb_i^k$ and $\b_i^k$ are parallel: 
\beq \label{eq:tildeb}
\tildeb_i^k = \frac{1}{1-\gamma_i^k \a_i^\hop \b_i^k}  \cdot   \b_i^k. 
\eeq 
It follows from \eqref{eq:tildeb} that 
  \beq \label{eq:SMformula3}
  \a_i^{\hop} \tildeb_i^k =  \frac{  \a_i^{\hop}  \b_i^k}{1-\gamma_i^k \a_i^\hop \b_i}
  \eeq  
  and hence inverting this relation yields: 
 \beq  
 \frac{1}{ \a_i^{\hop} \tildeb_i^{\iter}}  = \frac{1}{\a_i^\hop \b_i^{\iter}} -  \gamma_i^{\iter} . \label{eq:gamma_i_sol1_apu_2} 
 \eeq 
Again, since  $\tildeb_i^k \propto \b_i^k$  as shown in \eqref{eq:tildeb},  the following must hold:
\begin{align}
 \frac{[\tildeb_i^{\iter}]^{\hop} \SCM \tildeb_i^{\iter}}{  (\a_i^{\hop} \tildeb_i^{\iter})^2} &=  \frac{ [\b_i^{\iter}] ^{\hop} \SCM \b_i^{\iter}}{  (\a_i^{\hop} \b_i^{\iter})^2} . 
  \label{eq:gamma_i_sol1_apu_1} 
\end{align} 
Then  substituting  \eqref{eq:gamma_i_sol1_apu_1} and \eqref{eq:gamma_i_sol1_apu_2}  into \eqref{eq:gamma_i_update} allows us to write the root in an equivalent form 
\begin{align} \label{eq:60}
\tilde{\gamma_i}(\gam^k)
&= \gamma_i^{\iter}  + \sqrt{\frac{ [\b_i^{\iter}]^{\hop} \SCM \b_i^{\iter}}{ \big(\a_i^\hop \b_i^{\iter} \big)^3 }}  -   \frac{1}{\a_i^\hop \b_i^{\iter}} . 
\end{align}
If  $\tilde \gamma_i \ge 0$, the unconstrained minimizer is feasible (optimal), while if $\tilde \gamma_i < 0$, then $\nabla_{\gamma_i} \tilde f_i(\gamma_i) > 0$ for all $\gamma_i > \tilde \gamma_i$, in particular for all $\gamma_i \ge 0$. Thus the solution is 
$\hat \gamma_i(\gam^k)= \max( \tilde{\gamma_i}(\gam^k),0)$ as given in  \eqref{eq:SCA_optim_gammai}. 

\end{document}